\documentclass[a4paper,UKenglish]{lipics}
 
\usepackage{microtype}
\usepackage{fancyhdr,lastpage,latexsym,amsmath,amsfonts,lmodern,dsfont,stmaryrd,amssymb,complexity,tikz,enumitem}
\usetikzlibrary{positioning}
\usetikzlibrary{calc}
\usetikzlibrary{shapes}

\tikzstyle{every node}=[font=\scriptsize]

\def\centerarc[#1](#2)(#3:#4:#5)
{ \draw[line width=0.6pt,#1] ($(#2)+({#5*cos(#3)},{#5*sin(#3)})$) arc (#3:#4:#5); }

\def\line(#1:#2:#3)
{ \draw[line width=0.6pt] (#1,#2) -- (#1+#3,#2); }

\def\vlinetikz[#1](#2:#3:#4)
{ \draw[line width=0.6pt,#1] (#2,#3) -- (#2,#3+#4); }

\def\beam[#1](#2:#3:#4:#5:#6)
{ \draw[line width=0.6pt,#1] ({#5*cos(#6)+#2},{#5*sin(#6)+#3}) -- ({#4*cos(#6)+#2},{#4*sin(#6)+#3});  }

\def\line(#1:#2:#3:#4)
{ \fill [white] (#1,#2) rectangle (#1+#3,#2+#4); }

\newcommand{\isomorph}{\cong}
\newcommand{\probGI}{\text{GI}}
\newcommand{\argmin}{\operatornamewithlimits{argmin}}

\newcommand{\circlecover}{\ensuremath{\mathtt{cc}}}
\newcommand{\overlap}{\ensuremath{\mathtt{ov}}}
\newcommand{\disjoint}{\ensuremath{\mathtt{di}}}
\newcommand{\contains}{\ensuremath{\mathtt{cs}}}
\newcommand{\contained}{\ensuremath{\mathtt{cd}}}

\title{Deciding Circular-Arc Graph Isomorphism in Parameterized Logspace}
\author[1]{Maurice Chandoo}
\affil[1]{Leibniz Universität Hannover, Theoretical Computer Science,\\
  Appelstr. 4, 30167 Hannover, Germany\\
  \texttt{chandoo@thi.uni-hannover.de}}
\authorrunning{M. Chandoo}
\Copyright{Maurice Chandoo}
\subjclass{G.2.2 Graph Theory}
\keywords{graph isomorphism, canonical representation, parameterized algorithm}

\begin{document}

\maketitle

\begin{abstract}
        We compute a canonical circular-arc representation for a given circular-arc (CA) graph which implies solving the isomorphism and recognition problem for this class. To accomplish this we split the class of CA graphs into uniform and non-uniform ones and employ a generalized version of the argument given by Köbler et al.~(2013)  that has been used to show that the subclass of Helly CA graphs can be canonized in logspace. For uniform CA graphs our approach works in logspace and in addition to that Helly CA graphs are a strict subset of uniform CA graphs. Thus our result is a generalization of the canonization result for Helly CA graphs. In the non-uniform case a specific set $\Omega$ of ambiguous vertices arises. By choosing the parameter $k$ to be the cardinality of $\Omega$ this obstacle can be solved by brute force. This leads to an $\mathcal{O}( k + \log n)$ space algorithm to compute a canonical representation for non-uniform and therefore all CA graphs.         
 \end{abstract}
\section{Introduction} 
An arc is a connected set of points on the circle. A graph $G$ is called a CA graph if every vertex $v$ can be assigned to an arc $\rho(v)$ such that two vertices $u, v$ are adjacent iff their respective arcs intersect, i.e.~$\rho(u) \cap \rho(v) \neq \emptyset$.  We call such a (bijective) mapping $\rho$ a CA representation of $G$ and the set of arcs $\rho(G) = \left\{ \rho(v) \mid v \in V(G) \right\}$ a CA model. The recognition problem for CA graphs is to decide whether a given graph $G$ is a CA graph.  The canonical representation problem for CA graphs consists of computing a CA representation $\rho_G$ for a given CA graph $G$ with the additional canonicity constraint that whenever two CA graphs $G$ and $H$ are isomorphic the CA models $\rho_G(G)$ and $\rho_H(H)$ are identical. 

Similarly, a graph is an interval graph if every vertex can be assigned to an interval on a line such that two vertices share an edge iff their intervals intersect. It is easy to see that every interval graph is a CA graph since every interval model is a CA model.

The class of CA graphs started to gain attraction after a series of papers in the 1970's by Alan Tucker. However, there is still no known better upper bound for deciding CA graph isomorphism than for graph isomorphism in general even though considerable effort has been done to this end. There have been two claimed polynomial-time algorithms in \cite{wu}, \cite{hsu} which have been disproven in \cite{esc}, \cite{cur} respectively. For the subclass of interval graphs a linear-time algorithm for isomorphism has been described in \cite{lue}. A series of newer results show that canonical representations for interval graphs, proper CA graphs and Helly CA graphs (a superset of interval graphs) can be computed in logspace \cite{kob:intv,kkv12,kob:hca}. Furthermore, recognition and isomorphism for interval graphs is logspace-hard\cite{kob:intv} and these two hardness results carry over to the class of CA graphs; for recognition the reduction requires a little additional work.

Our main contribution is that we extend the argument used in \cite{kob:hca} to compute canonical representations for HCA graphs to all CA graphs.
We split the class of CA graphs into uniform and non-uniform ones and show that the mentioned argument can be applied in both cases using only $\mathcal{O}(k+\log n)$ space.
The parameter $k$ describes the cardinality of an obstacle set $\Omega$ that occurs only in the non-uniform case. This means $k=0$ in the uniform case and hence we also obtain a logspace algorithm for uniform CA graphs which are a superclass of HCA graphs. To the best of our knowledge this is the first non-trivial algorithm to decide isomorphism specifically for the class of CA graphs. 

This paper is structured as follows. In section 2 we define CA graphs along with their representations and recall the concept of normalized representations from \cite{hsu}. In section 3 we explain what we mean by flip trick and formalize this with the notions of flip sets and candidate functions. This idea has been used by \cite{mcc} to compute a CA representation for CA graphs in linear time and \cite{kob:hca} has modified it to compute canonical (Helly) CA representations for HCA graphs. In section 4 and 5 the flip trick is applied to uniform and non-uniform CA graphs respectively. 

    \begin{figure}[t]
        \centering
        \resizebox{9cm}{!}{%
        \begin{tikzpicture}[shorten >=1pt,auto,node distance=1.2cm,
  main node/.style={circle,draw}]

\newcommand*{\xoff}{0}%
\newcommand*{\yoff}{0}%

\newcommand*{\xoffca}{4.5}%
\newcommand*{\yoffca}{0}%

\newcommand*{\movea}{1.4}%
\newcommand*{\moveb}{0.8*sqrt(1/2)}%

\node[main node] (v1) at ({\xoff-\moveb},{\yoff+\moveb}) {2};
\node[main node] (v2) at ({\xoff-\moveb},{\yoff-\moveb}) {3};
\node[main node] (v3) at ({\xoff+\moveb},{\yoff-\moveb}) {4};
\node[main node] (v4) at ({\xoff+\moveb},{\yoff+\moveb}) {5};
\node[main node] (v5) at ({\xoff-2*\moveb},{\yoff}) {1};
\path[-]
(v1) edge (v4)
(v1) edge (v2)
(v1) edge (v5)
(v2) edge (v3)
(v3) edge (v4)
;
\node (M) at (\xoff,\yoff) {$G$};

\draw[line width=0.6pt,<-] (\xoff+3,\yoff) -- (\xoffca-3.4,\yoff);
\node (M) at (\xoff+2,\yoff+0.25) {$\rho(i) = A_i$};

\centerarc[](\xoffca,\yoffca)(90-10:180+10:0.9);
\centerarc[](\xoffca,\yoffca)(180-10:270+10:1);
\centerarc[](\xoffca,\yoffca)(270-10:360+10:0.9);
\centerarc[](\xoffca,\yoffca)(0-10:90+10:1);
\centerarc[](\xoffca,\yoffca)(125:150:1.01);

\node (a1) at ($({\xoffca+1.2*cos(45)},{\yoffca+1.2*sin(45)})$) {$A_5$};  
\node (a2) at ($({\xoffca+1.1*cos(45-90)},{\yoffca+1.1*sin(45-90)})$) {$A_4$};  
\node (a3) at ($({\xoffca+1.2*cos(45-180)},{\yoffca+1.2*sin(45-180)})$) {$A_3$};   
\node (a4) at ($({\xoffca+0.75*cos(-195)+0.1},{\yoffca+0.75*sin(-195)+0.1})$) {$A_2$};      
\node (a5) at ($({\xoffca+1.2*cos(45+90)-0.1},{\yoffca+1.2*sin(45+90)})$) {$A_1$};      
\node (M) at (\xoffca,\yoffca) {$\rho(G)$};

\end{tikzpicture}
        }
        \caption{A CA graph and representation}
        \label{fig:ca_intro}                
    \end{figure}
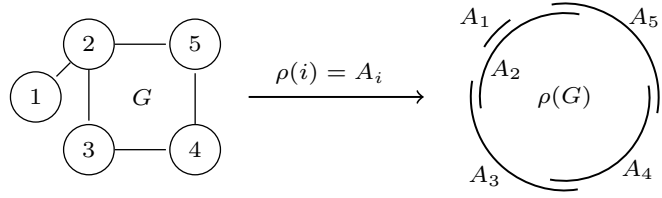  
    
\section{Preliminaries}
    Given two sets $A,B$ we say $A$ and $B$ intersect if $A \cap B \neq \emptyset$. We say $A,B$ overlap, in symbols $A \between B$, if $A \cap B, A \setminus B$ and $B \setminus A$ are non-empty. 
    Let  $A=(a_{u,v})_{u,v \in V(A)} , B =(b_{u,v})_{u,v \in V(B)}$ be two square matrices over vertex sets $V(A),V(B)$. We say $A$ and $B$ are isomorphic, in symbols $A \isomorph B$, if there exists a bijection $\pi \colon V(A) \rightarrow V(B)$ such that $a_{u,v} = b_{\pi(u),\pi(v)}$ for all $u,v \in V(A)$; $\pi$ is called an isomorphism. 
    For two graphs $G,H$ with adjacency matrices $A_G,A_H$ we say that $G \isomorph H$ if $A_G \isomorph A_H$.     
    We consider only undirected graphs without self-loops. A graph class $\mathcal{C}$ is a subset of all graphs which is closed under isomorphism, i.e.~if $G \in \mathcal{C}$ and $H \isomorph G$ then $H \in \mathcal{C}$. We define the graph isomorphism problem for a graph class $\mathcal{C}$ as $\probGI(\mathcal{C}) = \left\{ (G,H) \mid G,H \in \mathcal{C} \text{ and } G \isomorph H  \right\}$.  
    For a graph $G$ and a vertex $v \in V(G)$ we define the open neighborhood $N(v)$ of $v$ as the set of vertices that are adjacent to $v$ and the closed neighborhood $N[v] = N(v) \cup \{v\}$. For a subset of vertices $V' \subseteq V(G)$ we define the common neighborhood of $V'$ as $N[V'] = \bigcap_{v \in V'} N[v]$. We also write $N[u,v,\dots]$ instead of $N[\{u,v,\dots\}]$.  A vertex $v$ is called universal if $N[v] = V(G)$. For two vertices $u \neq v \in V(G)$ we say that $u$ and $v$ are twins if $N[u] = N[v]$.
    A twin class is an equivalence class induced by the twin relation. A graph is said to be without twins if every twin class has cardinality one.  
    For a graph $G$ and $S' \subseteq S \subseteq V(G)$ we define the exclusive neighborhood $N_{S}(S')$ 
    as all vertices $v \in V(G) \setminus S$ such that $v$ is adjacent to all vertices in $S'$ and to none in $S \setminus S'$.
    
    \begin{definition}[Label-independent]
Let $f$ be a function which maps graphs to a subset of subsets of vertices, i.e.~$f(G) \subseteq \mathcal{P}(V(G))$. We say $f$ is label-independent if for every pair of isomorphic graphs $G,H$ and all isomorphisms $\pi$ from $G$ to $H$ it holds that
$$ f(H) = \left\{ \pi(X) \mid X \in f(G) \right\} \text{ with } \pi(X) = \left\{ \pi(v) \mid v \in X \right\} $$ 
    \end{definition}
     
    \subsection{Circular-Arc Graphs and Representations}
    A CA model is a set of arcs $\mathcal{A} = \{A_1,\dots,A_n\}$ on the circle. Let $p \neq p'$ be two points on the circle. Then the arc $A$ specified by $[p,p']$ is given by the part of the circle that is traversed when starting from $p$ going in clockwise direction until $p'$ is reached. We say that $p$ is the left and $p'$ the right endpoint of $A$ and write $l(\cdot),r(\cdot)$ to denote the left and right endpoint of an arc in general. If $A = [p,p']$ then the arc obtained by swapping the endpoints $\overline{A} = [p',p]$ covers the opposite part of the circle. We say $\overline{A}$ is obtained by flipping $A$.    
    When considering a CA model with respect to its intersection structure only the relative position of the endpoints to each other matter. W.l.o.g.~all endpoints can be assumed to be pairwise different and no arc covers the full circle. Therefore a CA model $\mathcal{A}$ with $n$ arcs can be described as a unique string as follows. Pick an arbitrary arc $A \in \mathcal{A}$ and relabel the arcs with $1, \dots, n$ in order of appearance of their left endpoints when traversing the circle clockwise starting from the left endpoint of $A$. Then write down the endpoints in order of appearance when traversing the circle  clockwise starting from the left endpoint of the chosen arc $A$. Do this for every arc and pick the lexicographically smallest resulting string as representation for $\mathcal{A}$. For example, the smallest such string for the CA model in Fig.~\ref{fig:ca_intro} would result from choosing $A_1$ ($l(1),r(1),l(2),r(5),l(3),r(2),\dots$). In the following we identify $\mathcal{A}$ with its string representation. 
    
    Let $G$ be a graph and $\rho = (\mathcal{A},f)$ consists of a CA model $\mathcal{A}$ and a bijective mapping $f$ from the vertices of $G$ to the arcs in $\mathcal{A}$. Then $\rho$ is called a CA representation of $G$ if for all $u \neq v \in V(G)$ it holds that $\{u,v\} \in E(G) \Leftrightarrow f(u) \cap f(v) \neq \emptyset$. We write $\rho(x)$ to mean the arc $f(x)$ corresponding to the vertex $x$, $\rho(G)$ for the CA model $\mathcal{A}$ and for a subset $V' \subseteq V(G)$ let $\rho[V'] = \left\{ \rho(v) \mid v \in V' \right\}$. Given a set $X \subseteq V(G)$ we write $\rho^+(X)$ to denote $\cup_{v \in X} \rho(v)$. A graph is a CA graph if it has a CA representation. 
    
    We say a CA model $\mathcal{A}$ has a hole if there exists a point on the circle which is not contained in any arc in $\mathcal{A}$. Every such CA model can be understood as interval model by straightening the arcs. Therefore a graph is an interval graph if it admits a CA representation with a hole.
    A CA graph $G$ is called Helly (HCA graph) if it has a CA representation $\rho$ such that for all maxcliques, i.e.~inclusion-maximal cliques, $C$ in $G$ it holds that $\bigcap_{v \in C} \rho(v) \neq \emptyset$. Every interval model has the Helly property and therefore every interval graph is an HCA graph.         
    
    \subsection{Normalized Representation}
    In \cite{hsu} it was observed that the intersection type of two arcs $A \neq B$ can be one of the following five types: $A$ and $B$ are disjoint ($\disjoint$), $A$ is contained in $B$ ($\contained$), $A$ contains $B$ ($\contains$), $A$ and $B$ jointly cover the circle (circle cover $\circlecover$) or $A$ and $B$ overlap ($\overlap$) but do not  jointly cover the circle. 
    Using these types we can associate a matrix with every CA model. An intersection matrix is a square matrix with entries $\{ \disjoint, \overlap, \contains, \contained, \circlecover \}$. Given a CA model $\mathcal{A}$ we define its intersection matrix $\mu_\mathcal{A}$ such that $(\mu_{\mathcal{A}})_{a,b}$ reflects the intersection type of the arcs $a \neq b \in \mathcal{A}$. An intersection matrix $\mu$ is called a CA (interval) matrix if it is the intersection matrix of some CA (interval) model.
    
    When trying to construct a CA representation for a CA graph $G$ it is clear that whenever two vertices are non-adjacent their corresponding arcs must be disjoint. If two vertices $u,v$ are adjacent the intersection type of their corresponding arcs might be ambiguous. It would be convenient if the intersection type for every pair of vertices would be uniquely determined by $G$ itself. This can be achieved by associating a graph $G$ with an intersection matrix $\lambda_G$ called the neighborhood matrix which is defined for all $u \neq v \in V(G)$ as
    \[
    (\lambda_G)_{u,v} = 
    \begin{cases}
    \disjoint &, \text{if } \{ u,v \} \notin E(G) \\
    \contained &, \text{if } N[u] \subsetneq N[v] \\
    \contains &, \text{if } N[v] \subsetneq N[u] \\
    \circlecover &, \text{if } N[u] \between N[v] \text{ and } N[u] \cup N[v] = V(G) \\
    & \text{ and } \forall w \in N[u]\setminus N[v]: N[w] \subset N[u] \\
    & \text{ and } \forall w \in N[v]\setminus N[u]: N[w] \subset N[v] \\
    \overlap &, \text{otherwise}
    \end{cases}
    \]
    and the first case applies whose condition is satisfied. 
    
    Let $\mu$ be an intersection matrix over the vertex set $V$ and $\rho = (\mathcal{A},f)$ where $\mathcal{A}$ is a CA model and $f$ is a bijective mapping from $V$ to $\mathcal{A}$. We say $\rho$ is a CA representation of the matrix $\mu$ if $\mu$ is isomorphic to the intersection matrix of $\mathcal{A}$ via $f$ and denote the set of such CA representations with $\mathcal{N}(\mu)$. Then we say $\rho$ is a normalized CA representation of a graph $G$ if $\rho$ is a CA representation of the neighborhood matrix $\lambda_G$ of $G$. An example of a normalized representation can be seen in Fig.~\ref{fig:ca_intro}.  
    Let us denote the set of all normalized CA representations of $G$ with $\mathcal{N}(G) = \mathcal{N}(\lambda_G)$.
    
\begin{lemma}[\cite{hsu}]
    Every CA graph $G$ without twins and universal vertices has a normalized CA representation, that is $\mathcal{N}(G) \neq \emptyset$.    
\end{lemma}      
    For our purpose it suffices to consider only graphs without twins and universal vertices for the same reasons as in \cite{kob:hca}. The point is that a universal vertex can be removed from the graph and later added as arc which covers the whole circle in the representation. For each twin class an arbitrary representative vertex can be chosen and colored with the cardinality of its twin class; the other twins are removed.
    
    \begin{lemma}
        \label{lem:simpleonly}
        The canonical CA representation problem for CA graphs is logspace reducible to the canonical CA representation problem for colored CA graphs without twins and universal vertices.  
    \end{lemma}    
    Henceforth we assume every graph to be twin-free and without universal vertices and shall only consider normalized representations. 
    For two vertices $u \neq v$ in a graph $G$ we write $u \: \alpha \: v$ instead of $(\lambda_G)_{u,v} = \alpha$. For example, $u \: \contained \: v$ indicates that the arc of $u$ must be contained in the arc of $v$ for every (normalized) representation of $G$. 

\section{Flip Trick} 
    \begin{table}[b]
        \caption{Effects of flipping arcs in the intersection matrix}
        \label{tab:flip}
        \begin{center}
            \begin{tabular}{l | c c c c c}
                $\mu_{A,B}$ & \disjoint & \contained & \contains & \circlecover & \overlap \\
                \hline
                $\mu_{\bar{A},B}$ & \contains & \circlecover & \disjoint & \contained & \overlap \\
                $\mu_{A,\bar{B}}$ & \contained & \disjoint & \circlecover & \contains & \overlap \\
                $\mu_{\bar{A},\bar{B}}$ & \circlecover & \contains & \contained & \disjoint & \overlap 
            \end{tabular}
        \end{center}
    \end{table}   

Let $\mathcal{A}$ be a CA model and $X \subseteq \mathcal{A}$ is a subset of arcs to be flipped. We define the resulting CA model $\mathcal{A}^{(X)} = \left\{ \overline{A} \mid A \in X \right\} \cup \mathcal{A} \setminus X$. Consider a point $x$ on the circle and let $X$ be the set of arcs that contain this point. Then after flipping the arcs in $X$ no other arc contains the point $x$ and thus $\mathcal{A}^{(X)}$ has a hole and therefore must be an interval model. Let $\mu$ and $\mu^{(X)}$ be the intersection matrices of $\mathcal{A}$ and $\mathcal{A}^{(X)}$ respectively. It was observed in \cite{mcc} that the interval matrix $\mu^{(X)}$ can be easily computed using $\mu$ and $X$ as input via Table \ref{tab:flip}. With this the problem of computing a canonical CA representation for colored CA graphs can be reduced to the canonical interval representation problem for colored interval matrices, which can be solved in logspace\cite{kob:hca} (the colored part is not mentioned explicitly but can be easily incorporated into the proof by adding the colors to the leaves of the colored $\Delta$ tree).

The idea is that given a CA graph $G$ if we can compute a set of vertices $X$ as described above then we can obtain a canonical CA representation for $G$ by the following argument. The neighborhood matrix $\lambda_G$ of $G$ is a CA matrix and the matrix $\lambda_G^{(X)}$ must be an interval matrix. Compute a canonical interval representation for $\lambda_G^{(X)}$ and flip the arcs in $X$ back. This leads to a representation for $\lambda_G$ and thus $G$. 
The required set $X$ can be specified as follows.
\begin{definition}
    \label{prop:flipset}
    Let $G$ be a CA graph. Then a non-empty $X \subseteq V(G)$ is a flip set iff there exists a representation $\rho \in \mathcal{N}(G)$ and a point $x$ on the circle such that $v \in X \Leftrightarrow x \in \rho(v)$.     
\end{definition} 
In fact, for the argument to obtain a canonical representation to hold it is only required that $X$ is chosen such that $\lambda_G^{(X)}$ is an interval matrix. However, it can be shown that this is equivalent to the above definition, see Appendix.
Now, we can reframe the argument given in \cite{kob:hca} as follows.
    \begin{definition}[Candidate function]
        \label{def:candidatefunction}
        Let $\mathcal{C}$ be a subset of all CA graphs and $f$ is a function which maps graphs to a subset of subsets of their vertices, i.e.~$f(G) \subseteq \mathcal{P}(V(G))$. We call $f$ a candidate function for $\mathcal{C}$ if the following conditions hold: 
        \begin{enumerate}
            \item For every $G \in \mathcal{C}$ there exists an $X \in f(G)$ such that $X$ is a flip set
            \item $f$ is label-independent
        \end{enumerate}
    \end{definition}
    
    \begin{theorem}
        If $f$ is a candidate function for all CA graphs that can be computed in logspace then the canonical representation problem for CA graphs can be solved in logspace.
    \end{theorem}
    \begin{proof}
        Let $G$ be a graph. To decide the recognition problem for CA graphs observe that there is a flip set in $f(G)$ iff $G$ is a CA graph. To verify if a set $X \subseteq V(G)$ is a flip set one can check if $\lambda_G^{(X)}$ is an interval matrix by trying to compute an interval representation. 
        
        For the representation problem let $G$ be a CA graph. Let $F$ be the subset of $f(G)$ such that every $X \in F$ is a flip set. By the first condition it holds that $F$ is non-empty.
        For every $X \in F$ a CA representation $\rho_X$ of $G$ can be computed by the previous argument. We return a CA representation with the lexicographically smallest underlying model
        $$ \argmin_{\{ \rho_X \: | \: X \in F \}} \rho_X(G) $$
        as canonical CA representation.
        To see that this is indeed a canonical representation consider two isomorphic CA graphs $G,H$ with $F_G,F_H$ defined as $F$ for $G$ previously. Let $\pi$ be an isomorphism from $G$ to $H$ and $M_G = \{ \rho_X(G) \: | \: X \in F_G \}$ is the set of CA models induced by the flip sets $F_G$, similarly define $M_H$. Then canonicity follows by showing $M_G = M_H$. We show that  $M_G \subseteq M_H$ as the argument for the other direction is analogous. Let $\mathcal{A}$ be a model in $M_G$. Let $X$ be a flip set in $F_G$ which induces $\mathcal{A}$, i.e.~$\rho_X(G) = \mathcal{A}$. It must hold that $\pi(X) \in f(H)$ since $f$ is label-independent. From $G \isomorph H$ it follows that $\lambda_G \isomorph \lambda_H$ and therefore the interval matrices $\lambda_G^{(X)}$ and $\lambda_H^{(\pi(X))}$ are isomorphic meaning that $\pi(X) \in F_H$ since it is a flip set as well. 
        As the interval representations of both interval matrices have identical underlying models due to canonicity it follows that they remain so after flipping $X$ resp.~$\pi(X)$. Therefore it holds that $\mathcal{A} \in M_H$.
        
        This works in logspace since $f$ and the representation $\rho_X$ for every flip set $X$ can be computed in logspace.   
    \end{proof}    
So, we have reduced the problem of computing a canonical CA representation for CA graphs to the problem of computing a candidate function for CA graphs.

Let $\mathcal{C}$ and $\mathcal{C'}$ be two graph classes that partition all CA graphs. If $f_{\mathcal{C}},f_{\mathcal{C}'}$ are candidate functions for $\mathcal{C}, \mathcal{C'}$ respectively then $f(G) = f_{\mathcal{C}}(G) \cup f_{\mathcal{C}'}(G)$ is a candidate function for all CA graphs. That $f$ is label-independent follows from label-independent functions being closed under taking unions.
The crux here is that we do not need to be able to distinguish if a CA graph $G$ is in $\mathcal{C}$ or $\mathcal{C'}$. Hence, in the next two sections we consider two such classes that partition all CA graphs while avoid dealing with recognition of these two classes.

We complete this section by stating the candidate function used in \cite{kob:hca} to canonize HCA graphs and explain why it is a candidate function for this subclass of CA graphs. 
$$ f_{\text{HCA}}(G) =  \bigcup_{u,v \in V(G)}  \big\{ N[u,v]  \big\} $$

$f_{\text{HCA}}$ always returns at least one flip set for an HCA graph because all maxcliques in an HCA graph are flip sets due to the Helly property and there exists at least one maxclique in every HCA graph that can be characterized as the common neighborhood of two vertices\cite{kob:hca}. However, neither of these two properties hold for CA graphs in general.

It remains to argue that $f_{\text{HCA}}$ can be computed in logspace and is label-independent. Since the same arguments have to be made for the two candidate functions devised in the next sections we introduce a tool that facilitates this and demonstrate it for $f_{\text{HCA}}$.

\begin{definition}
    Let $\varphi$ be a first-order (FO) formula over graph structures with $k+1$ free variables. Then we define the function $f_\varphi$ for a graph $G$ as:
    $$ f_\varphi(G) = \bigcup\limits_{(v_1,\dots,v_k) \in V(G)^k}  
    \big\{ \left\{ u \in V(G)  \mid G \models \varphi(v_1,\dots,v_k,u) \right\}  \big\}
    $$
\end{definition}
\begin{lemma}
    \label{lem:easy}
    For every FO formula $\varphi$ over graph structures the function $f_\varphi$ is computable in logspace and label-independent.        
\end{lemma}
To compute $f_\varphi$ we can successively evaluate $\varphi$ and take the union of the results, which can be both done in logspace. To show that $f_\varphi$ is label-independent it suffices to apply structural induction to $\varphi$. The full argument for this can be found in the Appendix. 
Then $f_{\text{HCA}}$ is computed by the FO formula $\varphi(u,v,x)$ that states that $x \in N[u,v]$. By Lemma \ref{lem:easy} it follows that $f_{\text{HCA}}$ is logspace-computable and label-independent. 

\section{Uniform CA graphs}
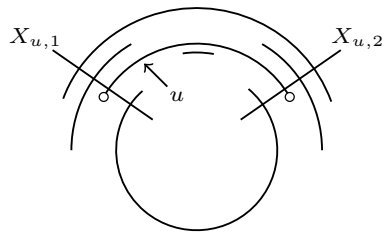
\begin{figure}[b]
    \centering
    \resizebox{5.5cm}{!}{%
        \begin{tikzpicture}[shorten >=1pt,auto,node distance=1.2cm,
  main node/.style={circle,draw}]

    \centerarc[](0,0)(30:150:1.2);
    \node[circle,draw,inner sep=0,minimum size=0.1cm,fill=white] (l1) at ($({1.2*cos(30)},{1.2*sin(30)})$) {};
    \node[circle,draw,inner sep=0,minimum size=0.1cm,fill=white] (r1) at ($({1.2*cos(150)},{1.2*sin(150)})$) {};    
    \centerarc[](0,0)(80:100:1.1);
    \centerarc[](0,0)(20:160:1.6);    
    \centerarc[](0,0)(30+20:150-360-20:0.9);    
    \centerarc[](0,0)(180:120:1.4);
    \centerarc[](0,0)(0:60:1.4);    

    \beam[](0:0:2:0.6:145)
    \node (xu1) at ($({2.2*cos(145)},{2.2*sin(145)})$) {$X_{u,1}$};
    \beam[](0:0:2:0.6:35)
    \node (xu2) at ($({2.2*cos(35)},{2.2*sin(35)})$) {$X_{u,2}$};
    
    \node (xu1) at ($({0.65*cos(110)},{0.65*sin(110)})$) {$u$};
    \beam[->]({0.65*cos(110)}:{0.65*sin(110)}:0.55:0.15:135)

\end{tikzpicture}
    }
    \caption{Uniform target flip sets}
    \label{fig:uniform_flipset}                    
\end{figure}  

The difficulty when trying to compute flip sets for CA graphs in general is that for a CA graph $G$ there might be different normalized CA representations such that a set of vertices shares a common point in one representation but not in an other one. We show a subset of CA graphs, namely the uniform CA graphs, where this issue does not occur therefore making it easy to compute flip sets.

For a CA graph $G$ consider an arbitrary vertex $u$. Looking at the neighbors of $u$ we can try to compute the flip sets $X_{u,1},X_{u,2}$ specified in Fig.~\ref{fig:uniform_flipset}. Both sets contain $u$ and all vertices that contain $u$ or form a circle cover with $u$. The vertices that overlap with $u$ belong to either $X_{u,1}$ or $X_{u,2}$ depending on the side they overlap from with $u$. Since we cannot determine left and right from the neighborhood matrix we want to express an equivalence relation $\sim_u$ which states that two vertices overlap from the same side with $u$. With the two equivalence classes induced by $\sim_u$ the two flip sets can be expressed. 

Given two vertices $x,y$ that both overlap with $u$ it is for instance easy to see that they must overlap from different sides with $u$ if they are disjoint. The only intersection type between $x,y$ for which the situation is not immediately clear is $\overlap$ as further distinctions are required. An $\overlap$-triangle is a set of three pairwise overlapping vertices. If $x$ and $y$ overlap then $x,y,u$ form such an $\overlap$-triangle. Consider the possible normalized representations for an $\overlap$-triangle. The three vertices can either all jointly cover the circle or be a set of overlapping intervals. In the first case they overlap pairwise but their overall intersection is empty thus we call this a non-Helly triangle and the second case an interval triangle. For an interval triangle there are three different possible representations up to reflection depending on which of the three vertices is placed in-between the other two. If $x,y,u$ is an $\overlap$-triangle then $x$ and $y$ overlap from the same side with $u$ iff $x,y,u$ form an interval triangle and $u$ is not in-between $x,y$.
We show that it is easy to derive this information in the case of a uniform CA graph $G$ as it does not depend on a representation of $G$. 

    \begin{definition}
        \label{def:deltag}
        Let $G$ be a graph. An $\overlap$-triangle $T$ is in the set $\Delta_G$ if the following holds:
        \begin{enumerate}
            \item $ \bigcup\limits_{v \in T} N[v] = V(G) $
            \item For all $x \in T$ it holds that if a vertex $v \in N_T(x)$ then $v \: \contained \: x$ 
        \end{enumerate}        
    \end{definition}        
    
    \begin{definition}[Uniform CA graph]
        A CA graph $G$ is uniform if for all $\overlap$-triangle $T$ in $G$ and $\rho \in \mathcal{N}(G)$ it holds that:
        $$   T \in \Delta_G \Rightarrow \rho[T] \text{ is a non-Helly triangle}$$
    \end{definition}    
    The idea behind the definition of $\Delta_G$ is that it captures the properties that an $\overlap$-triangle must satisfy in the graph if it can be represented as non-Helly triangle. The definition of uniform CA graphs guarantees us that an $\overlap$-triangle $T \in \Delta_G$ can never be represented as interval triangle.
    Now, we can show that for the class of uniform CA graphs the property of being a non-Helly triangle and the in-between predicate for interval triangles is invariant across all normalized representations.    
    \begin{lemma}
        \label{lem:deltag}
        Let $G$ be a uniform CA graph. Then the following statements are equivalent for every $\overlap$-triangle $T$:
        \begin{enumerate}
            \item $T \in \Delta_G$
            \item $\exists \rho \in \mathcal{N}(G): \rho[T]$ is a non-Helly triangle 
            \item $\forall \rho \in \mathcal{N}(G): \rho[T]$ is a non-Helly triangle
           \end{enumerate}
   \end{lemma}
   \begin{proof}
     This immediately follows from the definition of $\Delta_G$ and uniform CA graphs.
   \end{proof}    
As a contrasting example of a (non-uniform) CA graph for which being a non-Helly triangle depends on the representation consider the graph $G$ obtained by taking the complement of the disjoint union of three $K_2$'s (a triforce with a circle around the outer corners). Every edge in $G$ is an $\overlap$-entry in $\lambda_G$ and a CA model for $G$ is given by precisely two non-Helly triangles. The possible assignments of the vertices to the arcs that follow from the automorphisms of $G$ yield the different representations.      
    
    \begin{definition}
        Let $G$ be a graph and $T = \{u,v,w\}$ is an $\overlap$-triangle with $T \notin \Delta_G$. We say $v$ is in-between $u,w$ if at least one of the following holds:
        \begin{enumerate}
            \item $N_T(u), N_T(w) \neq \emptyset$
            \item $N_T(u,w) \neq \emptyset$ and there exists $z \in N_T(u,w)$ such that $\{u,w,z\} \in \Delta_G$
        \end{enumerate}
    \end{definition}    

    \begin{lemma}
        \label{lem:inbtw}
        Let $G$ be a uniform CA graph and $T=\{u,v,w\}$ is an $\overlap$-triangle with $T \notin \Delta_G$. Then the following statements are equivalent:
        \begin{enumerate}
            \item $v$ in-between $u,w$
            \item $\exists \rho \in \mathcal{N}(G): \rho(v) \subset \rho(u) \cup \rho(w) $
            \item $\forall \rho \in \mathcal{N}(G): \rho(v) \subset \rho(u) \cup \rho(w) $
        \end{enumerate}                
    \end{lemma}
    \begin{proof}
        ``$2 \Rightarrow 1$'': There exists $\rho \in \mathcal{N}(G)$ such that $\rho(v) \subset \rho(u) \cup \rho(w)$. For all $x \neq y \in T$ it must hold that $N[x] \setminus N[y] \neq \emptyset$ due to the fact that $x \: \overlap \: y$ implies $N[x] \between N[y]$. For this to be true $N_T(x)$ or $N_T(x,z)$ must be non-empty for $z \in T \setminus \{x,y\}$. It follows that  $N_T(u)$ and $N_T(w)$ or $N_T(u,w)$ must be non-empty since $N_T(v) = \emptyset$ due to $\rho(v) \subset \rho(u) \cup \rho(w)$. If $z \in N_T(u,w)$ then for $\rho(z)$ to intersect with $\rho(u)$ and $\rho(w)$ but not with $\rho(v)$ implies that $\rho[u,w,z]$ forms a non-Helly triangle and therefore $\{u,w,z\} \in \Delta_G$ by Lemma \ref{lem:deltag}.  
        \\``$1 \Rightarrow 3$'': Assume there exists a $\rho \in \mathcal{N}(G)$ such that $\rho(u) \subset \rho(v) \cup \rho(w)$. This implies that $N_T(u) = \emptyset$ and therefore there exists a $z \in N_T(u,w)$ such that $\{u,w,z\} \in \Delta_G$. By Lemma \ref{lem:deltag} it follows that $\rho[u,w,z]$ must be a non-Helly triangle and $\rho(z)$ must be disjoint from $\rho(v)$, contradiction.
        \\``$3 \Rightarrow 2$'': is clear.
    \end{proof}  
    Lemma \ref{lem:deltag} and \ref{lem:inbtw} state a fact of the form that if a property holds in one representation then it holds in all representations hence the name uniform CA graphs.

    The $\alpha$-neighborhood of a vertex $u$ in a graph $G$ is defined as $N^{\alpha}(u) = \left\{ v \in N(u) \mid u \: \alpha \: v   \right\}$ for $\alpha \in \{ \overlap, \contained, \contains, \circlecover \}$.  
    Now, we can define the aforementioned equivalence relation $\sim_u$ and state the candidate function for uniform CA graphs. 
    
    \begin{definition}
        Given a graph $G$ and vertex $u \in V(G)$ we define the relation $\sim_u$ on $N^{\overlap}(u)$ such that $x \sim_u y$ holds if one of the following applies:
        \begin{enumerate}
            \item $x = y$
            \item $x$ contains or is contained in $y$
            \item $x$ and $y$ overlap, $\{x,y,u\} \notin \Delta_G$ and $u$ is not in-between $x,y$
        \end{enumerate}            
    \end{definition}  
    
    \begin{lemma}
        \label{lem:sameside}
        Let $G$ be a uniform CA graph and $u \in V(G)$. Then $x \sim_u y$ holds iff $x$ and $y$ overlap from the same side with $u$ in every $\rho \in \mathcal{N}(G)$.            
    \end{lemma}
    \begin{proof}
        "$\Rightarrow$": If $x,y$ are in a contained/contains relation this is clear. For the third condition it holds that for all $\rho \in \mathcal{N}(G)$ $\rho[x,y,u]$ is an interval triangle and $x$ or $y$ is in-between the other two. It follows that $x,y$ overlap from the same side with $u$.     
        \\"$\Leftarrow$": If $\lambda_{x,y} \in \{\disjoint, \circlecover \}$ this is clear. If $x,y$ overlap then they either form a non-Helly triangle or $u$ is in-between $x,y$. In both cases $x,y$ overlap from different sides with $u$.
    \end{proof}    
    \begin{theorem}
        \label{thm:uniform}
        The following mapping is a candidate function for uniform CA graphs and can be computed in logspace:
        {\normalfont
        $$ f_{\text{U}}(G) =  
        \bigcup_{\substack{u \in V(G)\\ x \in N^{\overlap}(u) }} 
        \big\{ \{ u  \} \cup N^{\contained}(u) \cup N^{\circlecover}(u) \cup \{ y \in N^{\overlap}(u) | x \sim_u y \} \big\}      
        $$
        }
    \end{theorem}
    \begin{proof}
        To show that $f_{\text{U}}(G)$ is a candidate function we have to prove that for every uniform CA graph $G$ there always exists a flip set $X \in f_{\text{U}}(G)$ and that $f_{\text{U}}$ is label-independent. In fact, the even stronger claim holds that for all uniform CA graphs $G$ every set in $f_{\text{U}}(G)$ is a flip set. Let $X \in f_{\text{U}}(G)$ via some $u \in V(G)$ and $x \in N^{\overlap}(u)$. Then the set of vertices in $X$ correspond to one of the two flip sets shown in Fig.~\ref{fig:uniform_flipset}. The correctness for the subset of vertices in $X$ overlapping with $u$ follows from Lemma \ref{lem:sameside}.         
        
        To show that $f_{\text{U}}(G)$ can be computed in logspace and is label-independent we apply Lemma \ref{lem:easy}. We can rewrite $f_{\text{U}}$ as 
        $$ f_{\text{U}}(G) =  
        \bigcup\limits_{u,x \in V(G)} 
        \big\{ \left\{ z \in V(G) \mid G \models \varphi(u,x,z)  \right\} \big\}      
        $$
        with $G \models \varphi(u,x,z)$ iff  $x \in N^{\overlap}(u)$ and $z \in \{ u  \} \cup N^{\contained}(u) \cup N^{\circlecover}(u) \cup \{ y \in N^{\overlap}(u) | x \sim_u y \}$. It remains to check that the entries in the neighborhood matrix, $\alpha$-neighborhoods, exclusive neighborhoods, $\Delta_G$, in-between and $\sim_u$ can be expressed in FO logic.
    \end{proof}
    
    \begin{corollary}
        A canonical CA representation for uniform CA graphs can be computed in logspace.
    \end{corollary}

    \begin{theorem}
        Helly CA graphs are a strict subset of uniform CA graphs.
    \end{theorem}
    \begin{proof}
        First, we show "$\subseteq$" by contradiction. Assume there exists a Helly CA graph $G$ which is non-uniform. For $G$ to be non-uniform there must exist an $\overlap$-triangle $T \in \Delta_G$ and a representation $\rho \in \mathcal{N}(G)$ such that $T$ is represented as interval triangle in $\rho$. Let $T = \{u,v,w\}$ and assume w.l.o.g.~that $\rho(v) \subset \rho(u) \cup \rho(w)$ ($v$ is in-between $u,w$). It follows that $N[u] \cup N[w] = V(G)$. Since $\lambda_{u,w} \neq \circlecover$ there must be a $u' \in N[u] \setminus N[w]$ such that $N[u'] \setminus N[u] \neq  \emptyset$. This means there exists a $w' \in N[u',w] \setminus N[u]$. As $N[u'] \between N[w']$ it follows that $u'$ and $w'$ overlap. For $u'$ it must hold that it is either in $N_T(u)$ or $N_T(u,v)$. If it is in $N_T(u)$ then by the second condition of $\Delta_G$ it follows that $u'$ must be contained in $u$, contradiction. For the same reason $w'$ is in $N_T(v,w)$. It follows that $u',v,w'$ form an $\overlap$-triangle and must be represented as non-Helly triangle in $\rho$. This contradicts that $G$ is a Helly CA graph.
        
        To see that this inclusion is strict consider the graph $G$ obtained by taking a triangle $T$ and attaching a new vertex to each vertex in $T$ (also known as net graph). In every representation $\rho \in \mathcal{N}(G)$ it must hold that $T$ is represented as non-Helly triangle since $N_T(v) \neq \emptyset$ for all $v \in T$. For this reason $G$ cannot be a Helly or a non-uniform CA graph. 
    \end{proof}

\section{Non-Uniform CA graphs}
From the definition of uniform CA graphs it follows that a CA graph $G$ is non-uniform if there exists an $\overlap$-triangle $T$ in $\Delta_G$ and a $\rho \in \mathcal{N}(G)$ such that $T$ is represented as interval triangle in $\rho$. We call the pair $(T,\rho)$ a witness for the non-uniformity of $G$ and also say $G$ is non-uniform via $(T,\rho)$. Additionally, for such a witness pair $(T,\rho)$ we call $T$ maximal if there exists no $T' \neq T$ such that $G$ is non-uniform via $(T',\rho)$ and $\rho^+(T) \subset \rho^+(T')$. Such a maximal $T$ for a given $\rho$ must always exist. For the purpose of computing a candidate function for this class we can assume that for a given non-uniform CA graph $G$ we are supplied with an $\overlap$-triangle $T$ such that there exists a $\rho \in \mathcal{N}(G)$ with $(T,\rho)$ being a witness for $G$ and $T$ is maximal. This is justified by the fact that we can iterate over all $\overlap$-triangle $T \in \Delta_G$ trying to compute flip sets knowing that for at least one such $T$ these conditions are met. Additionally, we write $T$ as ordered triple $(u,v,w)$ to indicate that $v$ is in-between $u$ and $w$ in $\rho$. 

Let $G$ be non-uniform via $(T,\rho)$. Consider for a vertex $x \in V(G) \setminus T$ what the possible relations between $\rho(x)$ and $\rho^+(T)$ are. For instance, $\rho(x)$ cannot be disjoint from $\rho^+(T)$ because this implies that $x \notin \cup_{t \in T} N[t]$ and therefore $T \notin \Delta_G$. Also, $\rho(x)$ cannot contain $\rho^+(T)$ as this would mean that $x$ is a universal vertex. 
\begin{definition} 
    Let $G$ be a non-uniform CA graph via $(T,\rho)$. The set of normalized representations that agree with $\rho$ on $T$ is:
    $$ \mathcal{N}_{\rho}^T(G) = \left\{ \rho' \in \mathcal{N}(G) \mid \rho'[T] = \rho[T] \right\} $$
\end{definition}

\begin{definition}        
    Let $G$ be a non-uniform CA graph via $(T,\rho)$ and $\alpha \in \{\overlap, \circlecover, \contained \}$. We say $x \in V(G) \setminus T$ is an  $\alpha$-arc in $\rho'$ if
    $$ \rho'(x) \: \alpha \: \rho^+(T) $$
    for some $\rho' \in \mathcal{N}_\rho^T(G)$.  We call $x$ an unambiguous $\alpha$-arc if the above condition holds for all $\rho' \in \mathcal{N}_\rho^T(G)$.
\end{definition}  

    \begin{figure}[b]
        \centering
        \resizebox{12cm}{!}{%
            \begin{tikzpicture}[shorten >=1pt,auto,node distance=1.2cm,
  main node/.style={circle,draw}]

\newcommand*{\xoff}{0}%
\newcommand*{\yoff}{0}%

\newcommand*{\movea}{1.4}%
\newcommand*{\moveb}{0.8*sqrt(1/2)}%

\node (15) at ({\xoff},{\yoff}) {15};

\node (25) at ({\xoff-\moveb},{\yoff-\moveb}) {25};
\node (14) at ({\xoff+\moveb},{\yoff-\moveb}) {14};

\node (35) at ({\xoff-\moveb-\moveb},{\yoff-2*\moveb}) {35};
\node (24) at ({\xoff-\moveb+\moveb},{\yoff-2*\moveb}) {24};
\node (13) at ({\xoff+\moveb+\moveb},{\yoff-2*\moveb}) {13};

\node (45) at ({\xoff-\moveb-\moveb-\moveb},{\yoff-3*\moveb}) {45};
\node (34) at ({\xoff-\moveb},{\yoff-3*\moveb}) {34};
\node (23) at ({\xoff+\moveb},{\yoff-3*\moveb}) {23};
\node (12) at ({\xoff+\moveb+2*\moveb},{\yoff-3*\moveb}) {12};

\node (55) at ({\xoff-\moveb-\moveb-2*\moveb},{\yoff-4*\moveb}) {55};
\node (44) at ({\xoff-\moveb-\moveb},{\yoff-4*\moveb}) {44};
\node (33) at ({\xoff-\moveb+\moveb},{\yoff-4*\moveb}) {33};
\node (22) at ({\xoff+\moveb+\moveb},{\yoff-4*\moveb}) {22};
\node (11) at ({\xoff+\moveb+3*\moveb},{\yoff-4*\moveb}) {11};

\path[-]
(15) edge (25) 
(15) edge (14)

(25) edge (35) 
(25) edge (24) 
(14) edge (24) 
(14) edge (13) 

(35) edge (45)
(35) edge (34)
(24) edge (34)
(24) edge (23)
(13) edge (23)
(13) edge (12)

(45) edge (55)
(45) edge (44)
(34) edge (44)
(34) edge (33)
(23) edge (33)
(23) edge (22)
(12) edge (22)
(12) edge (11)
;

\newcommand*{\xoffa}{-9}%
\newcommand*{\yoffa}{-1}%
\newcommand*{\xlen}{3}%

\newcommand*{\ph}{1}%
\newcommand*{\pnh}{1.2}%
\newcommand*{\labelh}{0.9}%

\draw[line width=0.6pt] (\xoffa,\yoffa) -- (\xoffa+\xlen,\yoffa); 
\draw[line width=0.6pt] (\xoffa+\xlen/3,\yoffa+0.2) -- (\xoffa+\xlen/3+\xlen,\yoffa+0.2); 
\draw[line width=0.6pt] (\xoffa+2*\xlen/3,\yoffa-0.2) -- (\xoffa+2*\xlen/3+\xlen,\yoffa-0.2); 

\draw[line width=0.6pt,dotted] (\xoffa,\yoffa+\ph) -- (\xoffa,\yoffa-\pnh); 
\draw[line width=0.6pt,dotted] (\xoffa+\xlen/3,\yoffa+\ph) -- (\xoffa+\xlen/3,\yoffa-\pnh); 
\draw[line width=0.6pt,dotted] (\xoffa+2*\xlen/3,\yoffa+\ph) -- (\xoffa+2*\xlen/3,\yoffa-\pnh); 
\draw[line width=0.6pt,dotted] (\xoffa+3*\xlen/3,\yoffa+\ph) -- (\xoffa+3*\xlen/3,\yoffa-\pnh); 
\draw[line width=0.6pt,dotted] (\xoffa+4*\xlen/3,\yoffa+\ph) -- (\xoffa+4*\xlen/3,\yoffa-\pnh); 
\draw[line width=0.6pt,dotted] (\xoffa+5*\xlen/3,\yoffa+\ph) -- (\xoffa+5*\xlen/3,\yoffa-\pnh); 

\node (1) at (\xoffa+\xlen/6,\yoffa+\labelh) {1};
\node (1) at (\xoffa+\xlen/6+\xlen/3,\yoffa+\labelh) {2};
\node (1) at (\xoffa+\xlen/6+2*\xlen/3,\yoffa+\labelh) {3};
\node (1) at (\xoffa+\xlen/6+3*\xlen/3,\yoffa+\labelh) {4};
\node (1) at (\xoffa+\xlen/6+4*\xlen/3,\yoffa+\labelh) {5};

\node (u) at (\xoffa+0.2,\yoffa+0.15) {$u$};
\node (v) at (\xoffa+0.2+\xlen/3,\yoffa+0.35) {$v$};
\node (w) at (\xoffa-0.2+5*\xlen/3,\yoffa-0.05) {$w$};

\draw[line width=0.6pt] (\xoffa+1.5*\xlen/3,\yoffa-0.5) -- (\xoffa+4.5*\xlen/3,\yoffa-0.5);
\node[circle,draw,inner sep=0,minimum size=0.1cm,fill=white] (l1) at (\xoffa+1.5*\xlen/3,\yoffa-0.5) {};
\node[circle,draw,inner sep=0,minimum size=0.1cm,fill=white] (r1) at (\xoffa+4.5*\xlen/3,\yoffa-0.5) {};
\node[align=left] (xl) at (\xoffa+1.5*\xlen/3-0.08,\yoffa-0.5-0.32) {$l(x)$\\$\in 2$};
\node[align=left] (xr) at (\xoffa+4.5*\xlen/3-0.08,\yoffa-0.5-0.32) {$r(x)$\\$\in 5$};
\node[align=left] (x) at (\xoffa+1.5*\xlen/3+0.24,\yoffa-0.5+0.12) {$x$};

\end{tikzpicture}
        }
        \caption{Possible positions of the endpoints of $x$ relative to $T=\{u,v,w\}$}
        \label{fig:non_uniform_list}                    
    \end{figure}
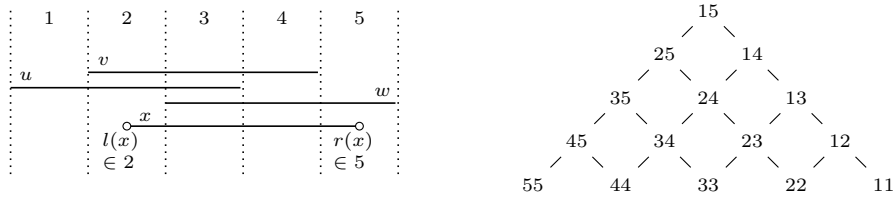  

    \begin{definition}
        Let $G$ be a graph and $T=(u,v,w) \in \Delta_G$. Then we define the following sets w.r.t.~$T$:
        \begin{alignat*}{2}
        &\Gamma_{\overlap,u} &&= \left\{ x \in V(G) \setminus T \mid                 
        x \: \overlap \:  u,v \: , \: x \: \disjoint \:  w   
        \right\} \\
        &\Gamma_{\overlap,w} &&= \left\{ x \in V(G) \setminus T \mid                 
        x \: \overlap \:  v,w \: , \: x \: \disjoint \:  u   
        \right\} \\        
        &\Gamma_\overlap &&= \Gamma_{\overlap,u} \cup \Gamma_{\overlap,w}\\                
        &\Gamma_\circlecover &&=  \left\{  x \in V(G) \setminus T \: \Big| 
        \begin{array}{l} 
        x \: \overlap \:  u,w  \: , \:                  
        x \: \disjoint \:  v  \text{ or}
        \\
        \exists a \in T: x \: \circlecover \: a                
        \end{array}                                    
        \right\} \\
        &\Gamma_\contained &&=  \left\{  x \in V(G) \setminus T \: \Big|  
        \begin{array}{l} 
        x \: \overlap \:  u,w  \: , \:                  
        x \: \contains \:  v  \text{ or}
        \\
        \exists a \in T: x \: \contained \: a                 
        \end{array}                                                               
        \right\}             \\
        &\Omega &&= \left\{ x \in V(G) \setminus T \mid x \: \overlap \: u,v,w \right\} \\
        \end{alignat*}
    \end{definition}                
    \begin{lemma}
        If $G$ is a non-uniform CA graph via $(T,\rho)$ such that $T$ is maximal then $T,\Gamma_\overlap,\Gamma_\circlecover,\Gamma_\contained,\Omega$ partition $V(G)$.       
    \end{lemma}
    \begin{proof}
        It is not hard to see that these sets do not overlap. To show that every vertex in $G$ belongs to one of these sets we need to check all possible positions of the endpoints of a vertex $x$ not in $T$ relative to $T$, consider Fig.~\ref{fig:non_uniform_list}. We know every vertex $x \notin T$ must be represented as $\alpha$-arc for some $\alpha \in \{\overlap,\contained,\circlecover\}$. For $\alpha \in \{\contained,\circlecover \}$ it must hold that both endpoints of $x$ must be in one of the intervals 1--5. The exemplary $x$ depicted in Fig.~\ref{fig:non_uniform_list} is a $\contained$-arc in the given representation and overlaps with $u,v,w$. If $x$ is a $\contained$-arc then the number of the interval in which its left endpoint is situated must be less than or equal that of its right endpoint (in our example $2 \leq 5$). If it is a $\circlecover$-arc then the right endpoint must come before the left. The graph on the right encodes all possible placements of the two endpoints and it can be verified case-by-case that an $\alpha$-arc for $\alpha \in \{\contained, \circlecover\}$ will occur in either $\Gamma_\alpha$ or $\Omega$. It remains to argue that an $\overlap$-arc can only have the intersection structure denoted by $\Gamma_{\overlap}$. W.l.o.g.~assume that $x$ is an $\overlap$-arc that overlaps from $u$'s side with $T$. It holds that the right endpoint of $x$ is in one of the five intervals and the left endpoint must be in none of these five intervals. If $r(x) \in 1$ then $x \in N_T(u)$ and by the second condition of $\Delta_G$ it must hold that $x$ is contained in $u$, contradiction. If $r(x) \in 2$ then $x$ overlaps with $u,v$ and is disjoint with $w$ and therefore $x \in  \Gamma_{\overlap,u}$. If $r(x) \in \{3,4\}$ then $T=\{u,v,w\}$ is not maximal since $\{x,v,w\} \in \Delta_G$. If $r(x) \in 5$ then it can be shown that there must be a circle cover entry in the neighborhood matrix for $x,w$ which contradicts that they must overlap. 
    \end{proof}

\begin{lemma}
      Let $G$ be a non-uniform CA graph via $(T,\rho)$ and $T$ is maximal. All vertices in $\Gamma_\alpha$ are unambiguous $\alpha$-arcs for $\alpha \in \{\overlap, \contained, \circlecover \}$.  
\end{lemma}
\begin{proof}
     The intersection structure of a vertex $x \in \Gamma_\alpha$ with $T$ dictates the positioning of the endpoints relative to $T$ in every $\rho' \in \mathcal{N}_\rho^T(G)$. It follows that this placement of the endpoints of $x$ must hold for all $\rho' \in \mathcal{N}_\rho^T(G)$ and therefore $x$ is an unambiguous $\alpha$-arc.
\end{proof}

For a vertex $x \in \Omega$ the possible placements of its endpoints to satisfy the intersection structure with $T$ can be one of the following four types $(\contained,14), (\contained,25), (\circlecover, 14), (\circlecover, 25)$. For example $x$ in Fig.~\ref{fig:non_uniform_list} is type $(\contained,25)$ and flipping $x$ would lead to type $(\circlecover,25)$.

Let $G$ be non-uniform via $(T,\rho)$ and $T=(u,v,w )$ is maximal. Then the two target flip sets $X_u,X_w$ we want to compute in the non-uniform case are immediately before the left endpoint and immediately after the right endpoint of $\rho^+(T)$ and must be of the form
$$ X_i = \Gamma_{\overlap,i} \cup \Gamma_{\circlecover} \cup \Omega' $$
for $i \in \{u,w\}$ and some subset $\Omega' \subseteq \Omega$. 

\begin{definition}
    Let $G$ be a non-uniform CA graph via $(T,\rho)$ and $T$ is maximal. We call a subset $\Omega' \subseteq \Omega$ $\circlecover$-realizable w.r.t.~$(T,\rho)$ if there exists a $\rho' \in \mathcal{N}_{\rho}^T(G)$ such that for all $x \in \Omega$ it holds that $x$ is a $\circlecover$-arc in $\rho'$ iff $x \in \Omega'$.
\end{definition}
In other words, a $\circlecover$-realizable set is a subset of $\Omega$ such that all of its vertices can be represented as $\circlecover$-arc in a normalized representation. By finding such a set we can construct two flip sets by adding the $\circlecover$-vertices in $\Gamma_\circlecover$ and one side of the $\overlap$-vertices as described above.    
Now, the challenge consists in finding such a $\circlecover$-realizable subset. A way to solve this is to parameterize our input by the cardinality of $\Omega$ and try all possibilities. Since $\Omega$ and its cardinality depend on the particular $\overlap$-triangle $T$ chosen, which we do not know a priori, we can use the following set which is a superset of every possible $\Omega$ and thus bounds the cardinality: 
$$ K_G = \left\{ x \in V(G) \mid \exists \{u,v,w\} \in \Delta_G \text{ s.t. } x \: \overlap \: u,v,w  \right\} $$

\begin{theorem}
    The following mapping is a candidate function for non-uniform CA graphs and can be computed in $\mathcal{O}(k + \log n)$ space for $k = |K_G|$:
    \normalfont
    $$ f_{\text{N}}(G) = \bigcup\limits_{(u,v,w) \in \Delta_G \atop \Omega' \subseteq \Omega  } \left\{ \Gamma_{\overlap,u} \cup \Gamma_{\circlecover} \cup \Omega' \right\}
    $$     
    where $\Gamma_{\overlap,u},\Gamma_{\circlecover},\Omega$ are taken w.r.t.~$T=(u,v,w)$.  
\end{theorem}
\begin{proof}
 To show that there always exists a flip set $X \in f_{\text{N}}(G)$ for a non-uniform CA graph $G$ we argue as follows. Let $G$ be non-uniform via $(T,\rho)$ and $T=(u,v,w)$ is maximal. Let $\Omega'$ be a $\circlecover$-realizable set w.r.t.~$(T,\rho)$, which must exist since $G$ is non-uniform. Then $X = \Gamma_{\overlap,u} \cup \Gamma_{\circlecover} \cup \Omega'$ w.r.t.~$T$ is one of the target flip sets described previously. This means there exists a $\rho \in \mathcal{N}(G)$ such that $X$ describes the set of arcs that contain a point right before the left endpoint of $\rho^+(T)$ or a point right after the right endpoint of $\rho^+(T)$.      
 
 To see that $f_{\text{N}}(G)$ is label-independent a formula $\varphi(\Omega',u,v,w,x)$ can be constructed that is true iff $x \in  \Gamma_{\overlap,u} \cup \Gamma_{\circlecover} \cup \Omega' $ where $\Omega'$ is a second-order set variable. Note, that $|\Omega| \leq |K_G|$. Therefore this works in $\mathcal{O}(k + \log n)$ space since one can iterate over all $2^k$ subsets of $\Omega$ using $k$ bits and then apply the argument in Lemma \ref{lem:easy} via $\varphi$ which requires additional $\mathcal{O}(\log n)$ space. 
\end{proof}

\subsection*{Conclusion}
We showed how to canonically, or in our terms label-independently, compute flip sets for CA graphs to acquire canonical CA representations. The properties of uniform CA graphs enable us to do this easily in logspace. In the case of non-uniform CA graphs, however, it seems that the $\circlecover$-realizable sets pose a non-trivial obstacle when trying to compute flip sets. The only simple remedy appears to be the proposed parameterization that enables us to use brute force. Changing the target flip sets does not seem to improve upon this situation. As a consequence, we suggest to investigate the space of $\circlecover$-realizable sets. Given the restricted structure of non-uniform CA graphs this could be a reasonable first step towards deciding isomorphism for CA graphs in polynomial time.

Additionally, in \cite{mcc} it was shown how to compute flip sets for CA graphs in linear time without the canonicity constraint. Can this be done in logspace as well? This would mean that recognition of CA graphs is logspace-complete.  

\subparagraph*{Acknowledgments}
We thank the anonymous reviewers for their helpful comments on earlier drafts of this paper.

\bibliography{ca}

\section{Appendix}
\begin{proof}[Proof of Lemma \ref{lem:easy}]
    $f_\varphi$ can be computed in logspace by successive evaluation of $\varphi$ and taking the union of the resulting sets. To prove that $f_\varphi$ is label-independent we use the following claim that can be verified by induction, see Lemma \ref{lem:folabel}.  For all FO formulas $\psi$ and isomorphic graphs $G,H$ it holds that
    $$ G \models \psi(v_1,\dots,v_k)  \iff H \models \psi(\pi(v_1),\dots,\pi(v_k))  $$
    for all isomorphisms $\pi$ from $G$ to $H$ and assignments $(v_1,\dots,v_k) \in V(G)^k$. Now, we must argue that $X \in f_\varphi(G)$ implies $\pi(X) \in f_\varphi(H)$. The other direction follows from a symmetrical argument. Let $X \in f_\varphi(G)$ then there exist $(v_1,\dots,v_k) \in V(G)^k$ such that $X = \left\{ u \in V(G) \mid G \models \varphi(v_1,\dots,v_k,u) \right\}$. This means 
    $$\pi(X) = \left\{ \pi(u) \mid u \in V(G) \text{ and } G \models \varphi(v_1,\dots,v_k,u)   \right\}$$
    We can replace $u \in V(G)$ with $u' \in V(H)$ and by the previous claim rewrite the above set as
    $$ \pi(X) = \left\{ u' \in V(H)  \mid H \models \varphi(\pi(v_1),\dots,\pi(v_k),u')   \right\} $$ 
    This concludes that $\pi(X) \in f_\varphi(H)$. 
\end{proof}    
    
\begin{lemma}
    \label{lem:folabel}
    For every FO formula $\varphi$ with $k$ free variables over graph structures and isomorphic graphs $G,H$ it holds that:
    $$ G \models \varphi(v_1,\dots,v_k)  \iff H \models \varphi(\pi(v_1),\dots,\pi(v_k))  $$
    for all isomorphisms $\pi$ from $G$ to $H$ and every assignment $(v_1,\dots,v_k) \in V(G)^k$.
\end{lemma}
\begin{proof}
    We show this by structural induction over FO formulas. For the base case $\varphi = E(x_i,x_j)$ the statement is clear, i.e.
    $$ G \models E(u,v)  \iff H \models E(\pi(u),\pi(v))  $$
    for all isomorphisms $\pi$ and $u,v \in V(G)$.
    For the inductive step we consider the boolean connectives $\neg, \wedge, \vee$ and the quantifiers $\exists, \forall$. Let $\varphi =  \neg \psi(x_1,\dots,x_k)$. 
    \begin{align*}
    & G \models \neg \psi(v_1,\dots,v_k) \\ \iff&
    G \not\models \psi(v_1,\dots,v_k)\\ \stackrel{\text{I.H.}}{\iff}&
    H \not\models \psi(\pi(v_1),\dots,\pi(v_k))\\ \iff&    
    H \models \neg \psi(\pi(v_1),\dots,\pi(v_k))                   
    \end{align*}           
    holds for all isomorphisms $\pi$ from $G$ to $H$ and $v_1,\dots,v_k \in V(G)$.
    For the cases $\wedge$ and $\vee$ this is similar. Let $\varphi = \exists x \psi(x,x_1,\dots,x_k)$. 
    \begin{align*}
    & G \models \exists x \psi(x,v_1,\dots,v_k) \\ \iff&
    \text{there exists } v \in V(G) \text{ such that } G \models \psi(v,v_1,\dots,v_k)  \\ \stackrel{\text{I.H.}}{\iff}&
    H \models \psi(\pi(v),\pi(v_1),\dots,\pi(v_k)) \\ \iff&
    H \models \exists x \psi(x,\pi(v_1),\dots,\pi(v_k)) 
    \end{align*} 
    for all isomorphisms $\pi$ from $G$ to $H$ and $v_1,\dots,v_k \in V(G)$.
    A similar argument holds for the $\forall$-case.                   
\end{proof}

\begin{lemma}
    \label{lem:flipseteq}
    Given a simple CA graph $G$ it holds that $X \subseteq V(G)$ is a flip set iff $\lambda_G^{(X)}$ is an interval matrix.
\end{lemma}
\begin{proof}
    If $X$ is a flip set then there exists a $\rho \in \mathcal{N}(G)$ and a point $x$ on the circle such that exactly all vertices in $X$ have the point $x$ in common in $\rho$. The resulting CA representation $\rho^{(X)}$ after flipping these arcs is an interval representation since $\rho^{(X)}(G)$ must have a hole at $x$. Since $\lambda_G^{(X)}$ must be isomorphic to the intersection matrix of $\rho^{(X)}(G)$ via $\rho$ it follows that it is an interval matrix.
    
    For the other direction let $X \subseteq V(G)$ such that $\lambda_G^{(X)}$ is an interval matrix. To reach a contradiction assume that $X$ is not a flip set. Then for every $\rho \in \mathcal{N}(G)$ and all points $x$ on the circle it holds that there exists a $v \in V(G)$ such that either $v \in X$ or $x \in \rho(v)$. Since $\lambda_G^{(X)}$ is an interval matrix it holds that there exists a $\rho_1 \in \mathcal{N}(\lambda_G^{(X)})$ such that $\rho_1$ has a hole. By flipping the arcs in $X$ in $\rho_1$ we acquire the representation $\rho_1^{(X)} \in \mathcal{N}(\lambda_G)$ for $G$. It holds that for every point $x$ on the circle there exists a $v \in V(G)$ such that $v \in X$ and $x \notin \rho_1^{(X)}(v)$ or $v \notin X$ and $x \in \rho_1^{(X)}(v)$. It follows that after flipping the set of arcs $X$ in $\rho_1^{(X)}$ back that no hole can exists which contradicts that $\rho_1$ has been an interval representation.
\end{proof}

\end{document}